\documentclass[hidelinks,reprint,aps,showpacs,superscriptaddress,nofootinbib]{revtex4-1}
\usepackage{changepage}

\usepackage{color}  %
\RequirePackage[dvipsnames,usenames]{xcolor}

\usepackage{graphicx}
\usepackage{amsmath} %
\usepackage{amssymb}  %
\usepackage{amsthm}

\usepackage{mathrsfs} %
\usepackage{stmaryrd} 
\usepackage{txfonts} %
\usepackage{lmodern}
\usepackage{mathtools, nccmath}
\usepackage{natbib}
\usepackage{enumerate}

\usepackage{multirow}

\usepackage{environ}
\NewEnviron{killcontents}{}
\usepackage[bottom]{footmisc}%
\usepackage{mathrsfs} %
\usepackage{stmaryrd} 
\usepackage{nicefrac}

\usepackage{enumerate}
\usepackage{appendix} 
\usepackage{enumitem}
\usepackage{hyperref}
\usepackage{comment}
\usepackage[capitalize]{cleveref}
\usepackage{thmtools,thm-restate}

\usepackage{tikz}
\usetikzlibrary{shapes.misc}
\tikzset{cross/.style={line width=2pt, rotate=20, cross out, draw=black, minimum size=10, inner sep=0pt, outer sep=0pt},cross/.default={5pt}}

\newcommand{\dhwc}[1] {}
\newcommand{\ak}[1] {}
\newcommand{\jao}[1] {}
\newcommand{\jaoa}[1] {}

\newif\ifnotes

\notesfalse

\def\shownotes{
\notestrue

\renewcommand{\dhwc}[1] {{\color{blue}{\bf{DHW: ##1}}}}
\renewcommand{\ak}[1] {{\color{ForestGreen}{{AK: ##1}}}}
\renewcommand{\jao}[1] {{\color{OrangeRed}{\bf{JAO: ##1}}}}
\renewcommand{\jaoa}[1] {{\color{Green}{\bf{Add: ##1}}}}

}

\newcommand{\nnr}{r}

\newcommand{\be}{\begin{equation}}
\newcommand{\bel}[1]{\begin{equation}\label{#1}}
\newcommand{\qe}{\end{equation}}
\newcommand{\ee}{\end{equation}}
\newcommand{\eeq}{\end{equation}}
\def\ba#1\ea{\begin{align}#1\end{align}}
\def\bann#1\eann{\begin{align*}#1\end{align*}}

\def\bal#1\eal{\begin{align}#1\end{align}}

\newcommand{\betight}{\begin{enumerate}[noitemsep,leftmargin=!,labelwidth=0em,labelindent=1em]
}
\newcommand{\eetight}{\end{enumerate}}

\newtheorem{theorem}{Theorem}

\newtheorem{corollary}[theorem]{Corollary}

\newtheorem{definition}{Definition}
\newtheorem{example}{Example}

\newtheorem{lemma}[theorem]{Lemma}

\newtheorem{proposition}[theorem]{Proposition}

\shownotes

\usepackage{blkarray}

\begin{document}

\title{Number of hidden states needed to physically implement a given conditional distribution}

 \author{Jeremy A. Owen}
\affiliation{Physics of Living Systems Group, Department of Physics, Massachusetts Institute of Technology, 400 Tech Square, Cambridge, MA 02139.}
   
 \author{Artemy Kolchinsky}
  \affiliation{Santa Fe Institute}

 \author{David H. Wolpert}
\altaffiliation{Massachusetts Institute of Technology}
  \affiliation{Santa Fe Institute}
 \affiliation{Arizona State University}
 \affiliation{{\tt http://davidwolpert.weebly.com}}

\date{\today}

\thanks{This is a post-peer-review version of an article published in \emph{New Journal of Physics}. The final, published version is available at \url{https://doi.org/10.1088/1367-2630/aaf81d}}

\begin{abstract}
We consider the problem of how to construct a physical process over a finite state space $X$
that applies some desired conditional distribution $P$ to initial states to produce final states. 
This problem 
arises often in the thermodynamics of computation and nonequilibrium statistical physics more generally (e.g., when designing processes to implement some desired computation, feedback controller, or Maxwell demon). It was previously known that some conditional distributions cannot be implemented using {\textit{any}} master equation that involves just the
states in $X$. However, here we show that any conditional distribution $P$ {can} in fact
be implemented---if additional ``hidden'' states not in $X$ are available.
Moreover, we show that is always possible to implement $P$ 
in a thermodynamically reversible manner. We then investigate 
a novel cost of the physical resources needed to implement a given distribution $P$: the 
minimal number of hidden states needed to do so.
We calculate this cost exactly for the special case where $P$ represents a single-valued function, and provide
an upper bound for the general case,
in terms of the nonnegative rank of $P$. These results show that
having access to one extra binary degree of freedom, thus doubling the total number of states, is sufficient to implement any $P$ with a master equation in a thermodynamically reversible way, if there are no constraints on the allowed form of the master equation. (Such constraints can greatly increase the minimal needed number of hidden states.) Our results also imply that for certain $P$ that {can} be implemented \textit{without} hidden states, having hidden states permits an implementation that generates less heat. \end{abstract}
\maketitle

\section{Introduction}
\label{sec:introduction}

Master equation dynamics over a discrete state space play a fundamental role in nonequilibrium statistical physics and stochastic thermodynamics~\cite{van2013stochastic,van_den_broeck_ensemble_2015,seifert2012stochastic}, 
and are used to model a wide variety of physical systems. Such dynamics can arise after coarse-graining the continuous phase space of a
 Hamiltonian system coupled to a heat bath~\cite{van_kampen_stochastic_1981,givon2004extracting,gaspard2006hamiltonian}, or as semiclassical approximations of the evolution of an open quantum system with discrete states~\cite{van2013stochastic,esposito2010finite}.

The linearity of the master equation implies a linear relationship between the distribution
over the system's states 
at the initial time $t=0$ 
and some final time (here taken to be $t=1$ without loss of generality). 
This is true even if the
master equation is time-inhomogeneous (i.e.,~non-autonomous). We can express this relation as
\[
p(1) = Pp(0)
\]  
where $p(0)$ and $p(1)$ are column vectors whose entries sum to one, representing probability distributions at times $t=0$
and $t=1$ respectively. If our system has $n$ states, the matrix $P$ is an $n \times n$ stochastic matrix, representing the conditional distribution of the system state at the final time given its state at the initial time. The entry $P_{ij}$ is the probability that the system is in state $i$ at the final time given that it was in state $j$ at the initial time.

The problem of characterizing the set of all stochastic matrices $P$ that can arise this way from a (time-inhomogeneous) continuous-time master equation is known as the \emph{embedding problem} for Markov chains~\cite{elfving_zur_1937,kingman_imbedding_1962,frydman_total_1979,fuglede_imbedding_1988,lencastre2016empirical,jia2016solution}, which can be viewed as the classical and time-inhomogeneous analogue of the \emph{Markovianity problem} studied in quantum information \cite{cubitt2012complexity, bausch2016complexity}. 

Despite the ubiquity of master equation dynamics in models of physical
systems, it turns out that many $P$ are  \emph{non-embeddable}, i.e., they cannot be implemented with any master equation. As we discuss below, examples of non-embeddable matrices include common operations such as bit erasure and the bit flip (a logical $\mathsf{NOT}$),
\begin{equation}
P_\text{erase} =
\begin{bmatrix}
1& 1\\
0 &0
\end{bmatrix},
\qquad
P_\text{flip} = \begin{bmatrix}
0& 1\\
1 &0
\end{bmatrix}.
\label{eq:peraseandswap}
\end{equation}

While neither of these operations are embeddable, there is a crucial difference between them. Bit erasure is infinitesimally close to being embeddable, meaning that it can be implemented by a master equation if one is willing to tolerate some finite, but arbitrarily small, probability of error. We 
call such matrices \textbf{limit-embeddable}.

On the other hand, as we show below, many stochastic matrices $P$, such as the bit flip, are not even limit-embeddable.
In light of this, suppose we encounter a physical system with $n$ ``visible states'' that we know obeys a 
continuous-time master equation, but is 
observed to evolve according to a $P$ that is not even limit-embeddable between $t=0$ and $t=1$.  
In such a situation we can conclude that
the master equation dynamics must actually take place over some larger state space, including some unseen \textbf{hidden
states} in addition to the $n$ visible states.

Alternatively, we can imagine that instead of observing some stochastic dynamics,
we attempt to \emph{build} a system that carries out a specified $P$ using master equation dynamics.
This is the challenge we might face, for example, if $P$ is the update function of the logical state of a computer we wish to build.
In this scenario, the minimal number of hidden states needed to implement
$P$ with master equation dynamics
emerges as a fundamental ``state space cost" of the computation
$P$. 
(See also~\cite{timestepspaper}.)

In fact, by adding hidden states one can construct master equations 
that meet more desiderata than just implementing a given $P$.
Specifically, below we show that for any given $P$ and initial distribution $p(0)$, one can construct a master equation
that implements $P$ on $p(0)$ while being arbitrarily close to thermodynamically
reversible.\footnote{$p(0)$ must be specified in addition to $P$ since
in general, the same master equation dynamics run on
different initial distributions $p(0)$ will implement the same $P$, but with
different amounts of irreversible entropy production. See ~\cite{kolchinsky2016dependence}.}

In this paper, we establish upper bounds on the number of hidden states required to
implement any stochastic matrix $P$ using a master equation, and also establish bounds
when we require that $P$ be implemented in a thermodynamically reversible way. 
We do this using explicit constructions that show any stochastic matrix $P$
can be implemented by composing some appropriate set of fundamental transformations that we call \textbf{local relaxations}, defined in detail below. 

Our main result is that any $n \times n$ stochastic matrix $P$ can be implemented with no more than 
$\nnr - 1$ hidden states, where $\nnr$ is the nonnegative rank of $P$. 
Because $\nnr \le n$, this result implies 
a simple corollary that one additional binary degree of freedom (which doubles the number of available states) is sufficient to carry out \emph{any} $P$ in a thermodynamically reversible way. 
We also derive exact results for some particular kinds of $P$, including those representing single-valued functions, 
which show that the number of required hidden states can be much  smaller than $\nnr-1$.

Finally, we show that if $P$ can be implemented with some number of hidden states while  incurring some nonzero entropy production, then it can be implemented without any entropy production 
at the cost of at most one additional hidden state. 
Such results imply, for instance, that for a system coupled to a heat bath, adding hidden states can allow one to implement the same $P$ while generating less heat, as we illustrate in detail in an example below.

Note that our results 
assume complete freedom to use any master equation to implement a given $P$.  
However, in the real world
there will often be major constraints on the form of the master equation
that can be considered, e.g., due to 
known properties of an observed system we wish to model using a master equation, 
or due to limitations
on what kind of system we can build. An example of the former is if we
know that the system's Hamiltonian can only couple degrees of
freedom in certain restricted ways. An example of the
latter is if $P$ must be implemented using a digital circuit made out of separate gates.
In cases where there are such constraints on the form of the master equation, 
our results can be considered as upper
bounds on \textit{lower bounds} of the number of hidden states that are really
required. (See \cref{sec:future_work} below.)

Previous research has shown how to carry out arbitrary $P$ thermodynamically reversibly~\cite{turgut_relations_2009,maroney2009generalizing,wolpert2015extending,wolpert_baez_entropy.2016}.
Moreover, the constructions in those papers can all be formulated in terms of master equation
dynamics.\footnote{Specifically, 
the processes considered in those papers are all (in our terminology) ``limit-embeddable'', and so could 
be expressed using master equations.}
In addition, they all exploit what in our terminology are called hidden states. However,
none of those papers considered the issue of the minimal number of hidden states
needed to implement $P$, i.e., the state space cost of implementing $P$,
which is the focus of this paper.

In a companion paper~\cite{timestepspaper} we focus on the special case where the stochastic matrix
$P$ is ``single-valued'', i.e., it represents a deterministic function. 
It turns out that for that case at least,
there is a second kind of cost arising in master equations dynamics
that implement $P$, in addition to the state space cost
which is the focus of this paper. Roughly speaking, that second cost is the minimal number of times that the set of allowed state-to-state transitions changes (which in a physics context may correspond to raising or lowering infinite energy barriers between states).
This can be viewed as a ``timestep cost'' of implementing $P$. Interestingly, there is 
a tradeoff between the timestep cost of implementing any (single-valued) $P$ and 
the state space cost of implementing that $P$. 

This paper focuses on the general problem of bounding the state space cost of arbitrary (not necessarily single-valued) stochastic matrices and the relationship of this cost to thermodynamic reversibility. In the next section we provide relevant background. Then in \cref{sec:def}, we define what it means for a master equation to implement a stochastic matrix $P$ to arbitrary precision. 
In \cref{sec:local_rel}, we define ``local relaxations", which are the building blocks of all our constructions. We present our main results in \cref{sec:embedding-with-aux-states}. 
We also investigate
how to extend our framework beyond finite state spaces to 
countably infinite state spaces, for the special case where $P$
is a single-valued map over $X$. Since this topic is a bit different
from the main focus of the paper, it can be found in \cref{app:infinite}.
The other appendices contain proofs that are not in the main text.

\section{Background}

\subsection{Master equations}

Consider a physical system with a finite state space $X$ of size $n$. 
We write the probability distribution of the state at time $t$ as $p(t)$, where $p_i(t)$ is the probability that the system is in state $i$ at time $t$. We suppose that $p(t)$ evolves as a time-inhomogeneous continuous-time Markov chain (CTMC), 
\begin{equation}
\frac{dp(t)}{dt} = M(t) p(t) \,.
\label{eq:ctmc_def}
\end{equation}
The elements in each column of the rate matrix $M(t)$ must sum to zero and its off-diagonal entries are positive. The entry $M_{ij}(t)$ is the transition rate from state $j$ to state $i$ at time $t$. Eq. \eqref{eq:ctmc_def} is commonly referred to as ``the master equation".


For any CTMC, we can relate the distributions at the initial time $t$ and some later time $t'$ by a linear map
\begin{equation}
p(t') = T_M(t,t') p(t) \,,
\end{equation}
where $T_M(t,t')$ 
is known as a \emph{transition matrix}, and equals the time-ordered exponential of $M(t)$.
Where $M(t)$ is clear from context, for notational convenience we will sometimes write $T_M(t,t')$ simply as $T(t,t')$. 
Note that if $M(t) = M$ is constant, then $T(t,t') = e^{(t'-t)M}$. Finally, note that we can rescale time arbitrarily by multiplying $M(t)$ by an appropriate constant. Therefore, without loss of generality we will take $t = 0$ and $t' = 1$ from now on.

\subsection{Embeddability}
\label{sec:embeddability}

Only some stochastic matrices $P$ are \emph{embeddable}, meaning that they can be written as $P = T_M(0,1)$ for some rate matrix $M(t)$. As an illustration of a non-embeddable matrix, 
consider a bit flip in a two state system, represented by $P_{\mathrm{flip}} = \left( \begin{smallmatrix} 0&1\\ 1&0 \end{smallmatrix} \right)$.  
Now the most general master equation for a two state system is
\begin{equation}
\dot{p}_1(t) = -r(t) p_1(t) + s(t) (1-p_1(t)) \,,
\label{bitflipmaster}
\end{equation}
where $p_1(t)$ indicates the probability that the bit is in state 1 at time $t$, and $r(t)$ and $s(t)$ indicate time-dependent rates of the $1\rightarrow 0$ and $0 \rightarrow 1$ transitions, respectively.  It can be shown that
there is no choice of non-negative functions
$r(t)$, $s(t)$ such that the solution of \cref{bitflipmaster} 
simultaneously satisfies $p_1(1) = 0$ when $p_1(0) = 1$
 and  $p_1(1) = 1$ when $p_1(0) = 0$, as required to implement $P_{\mathrm{flip}}$. To see why, 
 note that 
the space of distributions over two states is one-dimensional. 
The bit flip over 
a pair of states requires two 
distinct solutions to a differential equation, with different starting and ending conditions, to cross at some time, which is impossible (see \cref{fig:nonembed}).
(Note that this same argument does {not} apply to bit erasure.)

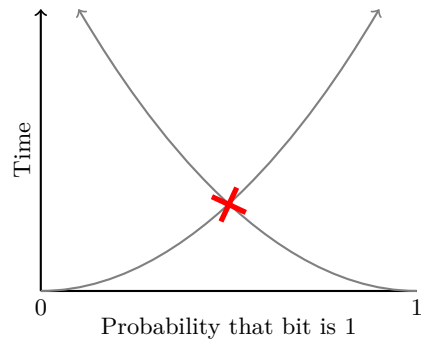
\begin{figure}
	\centering
	\begin{tikzpicture}[
        scale=5,
    ]
    \draw[thick] (0,0) -- (1,0) node[below] {1}; 
    \draw (0,0) node[below] {0};
    \draw[thick,->] (0,0) -- (0,0.75) ; 
    \node at (0.5, -0.1) {Probability that bit is 1};
    \node[rotate=90] at (-0.05,0.375) {Time};

    \draw[thick,gray,->] (0,0) parabola (0.9,0.75) ;
    \draw[thick,gray,->] (1.0,0.0) parabola (0.1, 0.75);

    \draw[thick] (0.5, 0.23) node[cross,red] {};
    
	\end{tikzpicture}
	\caption{The difference of two solutions to \cref{bitflipmaster} always decreases over time but it can never change sign.
\label{fig:nonembed}}
\end{figure}

Some broadly applicable necessary conditions for a stochastic matrix
to be embeddable are known.  For example,
by Liouville's formula, we have, 
\begin{align}
\det T_M(0,1) = \exp\left(\int_{0}^{1}  \operatorname{tr} M(t) \, dt\right) \,.
\label{eq:detid}
\end{align}
Since the exponential function is strictly positive, any embeddable matrix must have strictly positive
 determinant~\cite{goodman_intrinsic_1970,lencastre2016empirical}.  
An immediate corollary is that a
map like $P_\text{flip}$ (which has determinant $-1$) is not  embeddable. Indeed,
by the continuity of the determinant,  $P_\text{flip}$ is not even infinitesimally close
to an embeddable matrix. (In terms of our formal framework introduced below, 
$P_\text{flip}$ is not ``limit-embeddable''.)

It is also known that $P$ must obey $\det P \le \prod_i P_{ii}$ in order to be embeddable~\cite{goodman_intrinsic_1970, lencastre2016empirical}. 
Note that any (non-identity) permutation matrix has determinant that is either $1$ or $-1$, while the product of diagonal entries is $0$.
Thus, 
no permutation matrix is embeddable, or even 
infinitesimally close to an embeddable matrix.

These two necessary conditions on the determinant of a matrix
for it to be embeddable are not sufficient. Even if we restrict attention to the time-homogeneous case, where transition rates between states are constant, the problem of whether an arbitrary matrix
is embeddable is unsolved for matrices larger than $3 \times 3$. (On
the other hand, if one imposes the further constraint that the dynamics must obey detailed balance, the time-homogeneous embedding problem has been solved for all finite-size state spaces \cite{jia2016solution}.)


\subsection{Entropy production}

The field of stochastic thermodynamics~\cite{seifert2012stochastic} has developed a consistent thermodynamic interpretation of master equation dynamics, including 
a way to quantify \emph{entropy production} (EP), the overall increase of entropy in the system and coupled reservoirs.  For a physical system evolving according to a master equation like \cref{eq:ctmc_def}, we define the instantaneous rate of \emph{entropy production} at time $t$ to be~\cite{esposito2010three}
\begin{equation}
\dot{\Sigma}(p(t), M(t)) = \sum_{i,j} M_{ij}(t) p_j(t) \ln \frac{M_{ij}(t) p_j(t)}{M_{ji}(t) p_i(t)} \,
\label{eq:instEP}
\end{equation}

The \emph{total EP} over the time interval $[0,1]$, given some initial distribution $p(0)$, is 
\begin{equation}
\Sigma(p(0), M(t),[0,1]) = \int_{0}^{1} \dot{\Sigma}(p(t), M(t)) \; dt \,.
\end{equation}
The rate of EP is nonnegative, and therefore so is total EP.
A process that achieves 0 total EP is said to be \emph{thermodynamically reversible}.

If the system is coupled to a single reservoir, a heat bath at temperature $T$, then the total EP can be written as
\begin{equation}
\Sigma(p(0), M(t),[0,1]) = \left[ S(p(1)) - S(p(0)) \right] + \frac{1}{kT} \langle Q \rangle \,.
\label{eq:sigma_equals}
\end{equation}
where $S(\cdot)$ is Shannon entropy, $k$ is Boltzmann's constant, $\langle Q \rangle$ 
indicates the average heat transferred to the thermal reservoir 
over all trajectories of states from $t=0$ to $t=1$~\cite{esposito2010three}.

To use a well-known example from the thermodynamics of computation~\cite{landauer1961irreversibility,bennett1982thermodynamics}, suppose we wish to perform a bit erasure ($P_\text{erase}$), in which both initial states of a bit, $\{0,1\}$, get mapped to a single final value, $0$. Since this is a many-to-one map, it reduces the entropy of the device (e.g., a computer) in which it is performed. By the second law of thermodynamics, at least as much entropy must be produced in the environment. Often this occurs by the transfer of some amount of heat $\langle Q \rangle$ from the system into a heat bath at temperature $T$. In this case, assuming the initial probability over the states of the bit is uniform,
the second law yields $\langle Q \rangle \ge kT \ln 2$, which 
is the well-known bound on the amount of heat produced in bit erasure that
Landauer derived using semi-formal reasoning.\footnote{It is now understood that a logically irreversible process like bit erasure can be done in a thermodynamically \emph{reversible} manner~\cite{sagawa2014thermodynamic}, despite how many have interpreted Landauer's 
reasoning.} More generally,
an immediate consequence of \cref{eq:sigma_equals} and the non-negativity of EP is
\ba
\langle Q \rangle \ge kT [S(p(0)) - S(p(1))] \,,
\label{eq:minimal-heat}
\ea
where $p(0)$ and $p(1)$ are the distributions over system states at the beginning and end of the physical process under consideration. 


In the previous section, we discussed the embedding problem, which considers what $P$ can be implemented using a continuous-time master equation.  A related, thermodynamically motivated, question concerns which $P$ can be implemented using a master equation that achieves 0 total EP,
in the case where there is a single, isothermal heat bath. 
Generally, EP vanishes in such a situation if and only if $p(t)$ is an equilibrium distribution of $M(t)$ for all $t\in [0,1]$. Formally, this means that $M_{ij}(t)p_{j}(t)=M_{ji}(t)p_{i}(t)$ for all $i, j$ and 
$t\in [0,1]$. 
In general, this condition requires that the initial distribution 
$p(0)$ is an equilibrium distribution and that we
go to the ``quasistatic limit''  
where the parameters change infinitesimally slowly in relation to the relaxation time of the system~\cite{salamon_thermodynamic_1983,andresen_thermodynamics_1984,zulkowski2014optimal}. 
This limit can be achieved either by implementing $P$ increasingly slowly (in ``wall clock" time), and/or by making the transition rates (which control the relaxation time) increasingly large. 

While a given CTMC $M(t)$ may achieve zero (more accurately, arbitrarily small) 
total EP for some initial
distribution $p(0)$, in general it cannot generate zero EP for more than one such initial distribution. (See \cite{kolchinsky2016dependence} for a detailed investigation of the dependence of total
EP on the initial distribution.)
So our thermodynamically motivated question is whether a given $P$ has the property that for all initial distributions $p(0)$,
there is some associated set of $M(t)$ that implement $P$ with zero total EP, for that particular  initial distribution $p(0)$.

\section{Definitions\label{sec:def}}

\subsection{Limits of CTMCs}
\label{sec:limitofctmcs}

As discussed in the \cref{sec:introduction}, many common operations such as bit erasure are not embeddable (e.g.~no CTMC has $T(0,1)=P_\text{erase}$) but are nonetheless infinitesimally close to a stochastic matrix that is embeddable. This motivates the following definition:
\begin{definition}
\label{def:lim_embed}
A stochastic matrix  $P$ is \textbf{limit-embeddable} (LE) if for all $\epsilon > 0$, there is a rate matrix $M(t)$ such that $|T_M(0,1) - P| < \epsilon$.
\end{definition}

Our definition of limit embeddable makes no restrictions on the thermodynamic properties of the implementation of a given $P$.  Our next task is to formalize an additional restriction that reduces the set of all LE matrices to a subset that can be implemented in a thermodynamically reversible manner. Rather than focus on whether some $P$ can be implemented that way for one specific initial distribution, we focus on whether it can be implemented that way for \textit{any}
initial distribution. This motivates the following definition:
\begin{definition}
\label{def:qe}
A stochastic matrix $P$ is \textbf{quasistatically embeddable} (QE) if for all initial distributions $p$, all $\epsilon > 0$, and all $\delta > 0$, there is a rate matrix $M(t)$ with $|T_M(0,1) - P| < \epsilon$ and $\Sigma(p,M(t),[0,1]) < \delta$.
\end{definition}
\noindent We emphasize that 
the rate matrix $M(t)$ that implements a given $P$ must generally 
be specialized to the initial distribution $p$ in order to make $\Sigma(p,M(t),[0,1])$ arbitrarily small.



There are a few properties of QE matrices that are central to our analysis. The first one is the following lemma:
\begin{restatable}{lemma}{lemmaclosure}
The set of QE matrices is closed under multiplication.
\label{lemma:closure}
\end{restatable}
\noindent
It is impossible to implement any matrix without entropy production except in a limit. This is illustrated by the next proposition, which is proved using a result from~\cite{Saito2018SpeedLimit}:
\begin{restatable}{proposition}{EPboundprop}
	\label{prop:EPbound}
	If $M(t)$ is the rate matrix of a CTMC with $T(0,1) = P \neq I$, then for any initial distribution $p$,
	\[
	\Sigma(p, M(t),[0,1]) \geq -\frac{\left|Pp - p \right|_1^2}{2 \ln(\det P)} \,,
	\]
	where $\left| \cdot \right|_1$ is $\ell_1$ norm.\footnote{For two probability distributions $p$, $q$ over $n$ states, the $\ell_1$ norm is defined $|p-q|_1 := \sum_{i=1}^n |p_i - q_i|$.}
\end{restatable}

The inequality in \cref{prop:EPbound} (which may not be tight)
shows that only those $P$ with determinants very close to 0 can achieve small EP for arbitrary initial distributions.   
Recall that by \cref{eq:detid}, $\det P=\exp\left(\int_0^1\operatorname{tr} M(t) \, dt\right)$. 
This means that a small determinant can be achieved if the transition rates (the off-diagonal entries of $M(t)$) are very large (so that $\operatorname{tr} M(t)$ is very negative). Note that since we have scaled time out to fix the final time $t' = 1$, we are effectively measuring rates in units of $1/t'$. Therefore, the observation that the entries of $M(t)$ must be very large to achieve small $EP$ is consistent with the well-known fact that quasistatic processes are infinitely slow.
(See also our discussion of the quasistatic limit in the previous section.)

An important corollary of \cref{prop:EPbound} is that \emph{any} QE matrix must be singular. 
\begin{corollary}
\label{corr:QE-is-singular}
Any QE matrix, except the identity, has determinant zero.
\end{corollary}

These points are illustrated in the following example of a canonical QE process, bit erasure, implemented
with a quantum dot:
 \begin{example}
In the model of bit erasure described in~\cite{diana2013finite}, a classical bit is implemented
as a quantum dot, which can be either empty (state 0) or filled with an electron (state 1).  
The dot is brought into contact with a metallic lead at temperature $T$ which (for an
appropriate state of the dot) may transfer an electron into the dot or out of it, thereby changing
the state of the bit. 

At time $t$, the propensity of the lead to give or receive an electron is set by its chemical potential, indicated by $\mu(t)$, and the energy of an electron in the dot, indicated by $E(t)$.
Specifically,
let $p(t)$ indicate the two-dimensional vector of probabilities, with $p_0(t)$ and $p_1(t)$ being the probability of an empty and full dot, respectively.  These probabilities evolve according to the rate matrix~\cite{diana2013finite}
\begin{align}
M(t) = C
\begin{bmatrix}
 - w(t) & 1-w(t) \\
  w(t) &  -(1-w(t))
\end{bmatrix}  \,,
\label{eq:quantum-dot-deriv}
\end{align}
where $C$ sets the timescale of the exchange of electrons between the dot and the lead and 
$w(t)$ is the Fermi distribution of the lead,
\[
w(t) = \left[\exp((E(t)-\mu(t))/k_B T)+1\right]^{-1}\,.
\]
Using \cref{eq:ctmc_def}, \cref{eq:quantum-dot-deriv} and conservation of probability (i.e., $p_0(t) + p_1(t) = 1$), we can write

\begin{equation}\label{master}
\dot{p_1}(t) = C (w(t) - p_1(t))
\end{equation}

Suppose that $\mu(t)$ and $E(t)$ are chosen in a way that depends on $p_{1}(0)$,
so that $w(t) = (1-t) p_1(0) + t \delta$ for some constant $\delta$. In this case, \cref{master} can be explicitly solved for $p_1$:

\begin{eqnarray}
p_1(t) = w(t) + C^{-1}(p_1(0)-\delta)\left(1-e^{-Ct}\right).
\label{eq:10}
\end{eqnarray}
In the limit where $C \to \infty$,
$p_1(1) \to w(1) = \delta$, so the transition matrix $T(0,1)$ 
becomes
\[
T(0,1) = \begin{bmatrix}
 1-\delta & 1-\delta \\
  \delta &  \delta
\end{bmatrix}.
\] 
Furthermore, by \cref{eq:10}, in that limit ${p_1}(t) \to {w}(t)$, and so ${p_0}(t) \to 1 - {w}(t)$.
That means that for all $i, j$, $M_{ij}(t)p_{j}(t) \rightarrow M_{ji}(t)p_{i}(t)$.
In this limit, the total entropy production over the course of the process,
\[
\Sigma = \int_0^1 dt \, \sum_{i,j \in \{0, 1\}} M_{ij}(t) p_{j}(t) \ln\left[ \frac{M_{ij}(t)p_{j}(t)} {M_{ji}(t)p_{i}(t)}\right] \, ,
\]
approaches zero (we show this rigorously, in the context of proving a more general result, in \cref{LRisQE}). 

Since $\delta$ is arbitrary, this means that we can make $T(0,1)$ arbitrarily close to bit erasure ($\delta  = 0$),
while having arbitrarily small total EP. This establishes that bit erasure is QE.

Note that if we cannot control $C$ directly, we can still achieve
the same effect as the limit $C \to \infty$ by running the process with some fixed $C$ for longer and longer times (that is, by changing the endpoint from $t = 1$ to some $t > 1$). This demonstrates the equivalence of going to the limit of infinitely-long time versus infinitely-fast rates for achieving vanishing EP. 

\label{ex:1}
\end{example}

\subsection{Embedding with hidden states}
\label{sec:embedding-problem}

Many stochastic matrices $P$ are not QE. Our main result is that despite this, any $P$ is a (principal) submatrix of a larger matrix $\tilde{P}$ which \textit{is}  QE. We
formalize this property as follows:
\begin{definition}
\label{def:QEwithhiddenstates}
	An $n \times n$ stochastic matrix $P$ is \textbf{quasistatically embeddable with $m$ hidden states} if there exists some $(n+m) \times (n+m)$ matrix $\tilde{P}$ that is QE, and for all $i,j \in 1\dots n$, $\tilde{P}_{ij} = P_{ij}$ (i.e., $P$ is a principal submatrix of $\tilde{P}$).
\end{definition}

To understand the motivation for this definition, imagine that there is a matrix
$\tilde{P}$ over state space $Y$ that is QE, and that $P$ is a principal submatrix of $\tilde{P}$, corresponding to the subset of states $X \subseteq Y$. If at $t = 0$ the process that implements $\tilde{P}$ is in some state $i \in X$, then the distribution over the states in $X$ at the end of the process at $t=1$ will be exactly as specified by $P$.  
Furthermore, because  $\sum_{j\in X} P_{ji} = 1$, $\tilde{P}_{ji} = 0$ for any $i \in X$ and $j \not\in X$. This means that if the process is started on some $i\in X$, no probability can ``leak" out into the hidden states by the end of the process, although it may pass through them at intermediate times.

Mathematically, the hidden states in \cref{def:QEwithhiddenstates} are {additional} to the states whose evolution is controlled by $P$. However, there are multiple ways to map this mathematical structure into a particular physical system. In particular, hidden states can arise when $Y$ is
a set of microstates of a system and $X$ is an associated set of macrostates; we elaborate this point in Section \ref{sec:embedding-with-aux-states}.

Note that any $P$ that is QE with $m$ hidden states is also QE with $p$ hidden states for all $p \geq m$. Note as well that by \cref{lemma:closure}, if an $n \times n$
stochastic matrix $P$ can be factored into a product of $n \times n$ matrices, each
of which is QE with $m$ hidden states, then $P$  itself is QE with $m$ hidden states. More generally, 
the number of hidden states required to quasistatically embed a product 
of stochastic matrices is no more than the maximum number required for each of the matrices in that product.
In \cref{sec:embedding-with-aux-states}, we exploit this fact to derive upper bounds on the minimal number of hidden states
needed to quasistatically embed a given matrix.

The definition of QE may seem very strict, but from the perspective of the number of hidden states needed for embedding, it is only slightly more costly than limit embedding:
\begin{restatable}{proposition}{embedqeprop}
\label{embedqe1}
If a matrix $P$ is limit-embeddable, it is QE with at most one hidden state.
\end{restatable}
\noindent
This means the {minimal} number of hidden states required for limit-embedding vs.~quasistatic-embedding are at most one apart.


Note that it might be fruitful to consider a weaker definition of QE than \cref{def:qe}, in which a matrix would be considered ``QE'' if it is limit-embeddable while achieving vanishingly small EP for \emph{some} specific initial distribution $p$, 
rather than requiring that it is limit-embeddable 
that way for \textit{any} initial distribution.  With this definition,
the ``state space cost'' given by
the number of hidden states would then be a function of $P$ and the initial distribution $p$, not
just of $P$. Our definition can be seen as a worst-case version
of this weaker definition; we measure the ``state space cost'' of a stochastic matrix $P$ as the
maximal number of hidden states we might need if we were given some specific initial distribution $p$ and wanted to construct a master equation that implements
$P$ with no EP for that $p$.

In addition, we note that any matrix which is QE according to our stronger definition would also satisfy this weaker definition of QE. Thus, the upper bounds on the minimal number of hidden states required to quasistatically embed a given $P$ (according to \cref{def:qe}) which
we derive below are also upper bounds for the number of hidden states required under the alternative weaker definition of QE.  Note also that \cref{corr:QE-is-singular} holds under the weaker definition, as long as the desired initial distribution $p$ is not the fixed point of $P$.

\section{Local relaxations}
\label{sec:local_rel}


In \cref{ex:1} we showed that bit erasure is QE. In fact, bit erasure is part of a much larger family of stochastic matrices which can be shown to be QE, which we call ``local relaxations''.
These will serve as the ``building blocks'' of the constructions we use below:

\begin{definition}
A stochastic matrix $P$ is a \textbf{local relaxation} if there is some permutation matrix $Q$ (which may be the identity) such that:
\begin{enumerate}
\item $Q P Q^{-1}$ has a block diagonal structure,
\item Each block of $Q P Q^{-1}$ has rank one.
\end{enumerate}
\end{definition}
\noindent Note that the role of the permutation matrix $Q$ is to just rearrange rows and columns, i.e., to relabel states. 

Loosely speaking, when a local relaxation (LR) matrix is used to evolve a system, states 
of the system are grouped into different ``blocks'' between which no probability flows, while the states within each block relax to the same final, block-specific distribution. 


\begin{example}
The following $4 \times 4$ stochastic matrices are all local relaxations (for each matrix, we assume that $a,b,c,d$ are chosen to be nonnegative and that the sum of each column equals 1):
\[
A = \begin{bmatrix}
a & a & a & a \\
b & b & b & b \\
c & c & c & c \\
d & d & d & d
\end{bmatrix}
\enskip ; \enskip
B = \begin{bmatrix}
a & a & 0 & 0 \\
b & b & 0 & 0 \\
0 & 0 & c & c \\
0 & 0 & d & d
\end{bmatrix}
\enskip ; \enskip
C = \begin{bmatrix}
a & 0 & a & 0 \\
0 & c & 0 & c \\
b & 0 & b & 0 \\
0 & d & 0 & d
\end{bmatrix}
\]
The block structure in matrices $A$ and $B$ is immediately apparent.  Matrix $C$ is a local relaxation since $B = Q C Q^{-1}$, where $Q$ is the permutation matrix
$$Q = \begin{bmatrix} 1 & 0 & 0 & 0 \\ 0 & 0 & 1 & 0 \\ 0 & 1 & 0 & 0 \\ 0 & 0 & 0 & 1 \end{bmatrix} \,.$$
(In other words, $C$ can be arranged to have the block structure of $B$ by switching rows/columns 2 and 3.)
\end{example}

In \cref{LRisQE}, we prove the following result, which establishes that we can use LR
matrices as building blocks to construct QE matrices:

\begin{restatable}{proposition}{proprelaxQE}
	\label{prop:relaxQE}
	Any local relaxation is QE.
\end{restatable}

By performing one local relaxation followed by another one involving a different partition of $X$ into blocks, 
we can quasistatically implement matrices that are not themselves local relaxations. This means the converse of \cref{prop:relaxQE} is false---not every QE matrix is a local relaxation.

However, it \emph{is} the case that any $2\times 2$ QE matrix $P$ is a local relaxation. If $P$ is the identity, it is a local relaxation. If not, then by \cref{corr:QE-is-singular}, its determinant is zero, which implies it has rank one and so is a local relaxation (with a single block).

Since they are all QE, products of LR matrices have determinant zero (again, except for the identity). On the other hand, not all singular matrices are products of LRs. To see this, note that a (non-identity) product of LRs must always send at least two initial states to the same final distribution (since there must at least one LR in the product with a block of size larger than one). So such a product must have at least two identical columns. Therefore, a matrix like
\[
\begin{bmatrix}
1 & 0 & 1/2 \\
0 & 1 & 1/2 \\
0 & 0 & 0
\end{bmatrix}
\]	
is not a product of LRs, even though it is singular.

To establish the results in the next section, we will repeatedly make use of \cref{prop:relaxQE} and \cref{lemma:closure} to prove that a given matrix is QE by writing it as a finite product of local relaxations. 

One might conjecture that a matrix is QE only if it is a product of LRs. However, we do not establish this, and our results below do not rely on it.

\section{Upper bounds on minimal number of hidden states\label{sec:embedding-with-aux-states}}

\subsection{Single-valued maps}

We refer to a stochastic matrix that represents a deterministic function (i.e., all its entries are either 0 or 1) as a \textbf{single-valued map}.

To begin, note that if a single-valued map $P$ corresponds to an invertible function, then $P$ is a permutation and is not limit-embeddable, as discussed in 
\cref{sec:introduction}. 
If, on the other hand, a single-valued map $P$ corresponds to a non-invertible function, then it has determinant 0.  Moreover, we can use LR matrices to establish that any such $P$ is  limit-embeddable, and in fact QE:



\begin{proposition}
\label{prop:single-valued-singular}
Any single-valued map with determinant zero is QE. 
\end{proposition}
\begin{proof}
Consider a single-valued map $P$ over some finite set $X$ with determinant zero, which represents some non-invertible function $f: X\to X$. It is known that any non-invertible function $f:X\to X$ is a composition of finitely many idempotent $X \to X$ functions,\footnote{A function is \emph{idempotent} if applying it twice is the same as applying it once, $f \circ f= f$.} $f = f_N \circ \dots f_2 \circ f_1$~\cite{howie_subsemigroup_1966}. In our context, this means that $P$ can be represented as  a product of $N$ single-valued maps $P = \prod_k M^{(k)}$, each $M^{(k)}$ representing the idempotent function $f_k$.
Note that the image of an idempotent function consists entirely of fixed points, thus each $M^{(k)}$ can be written as
\[
M^{(k)}_{ij} = \begin{cases}
1 & \text{if $j \in f_k^{-1}(i)$ and $i = f_k(i)$} \\
0 & \text{otherwise}
\end{cases} \,,
\]
It can be verified that any such $M^{(k)}$ is a local relaxation, since it consists of blocks corresponding to the set of pre-images $\{f^{-1}(i) : i \in f(X)\}$, and each block $f^{-1}(i)$ contains identical columns: all 0s, except for a 1 in the row corresponding to $i$.  

 
The result then follows by applying first \cref{prop:relaxQE} and then \cref{lemma:closure}.
\end{proof}

While a single-valued map with nonzero determinant is not QE (by \cref{corr:QE-is-singular}), 
it can always be made QE by adding a single hidden state. Before establishing this for the 
general case,
we illustrate the basic proof technique with an example showing that
the bit flip is QE with one hidden state:

\begin{example}
The bit flip $P_{\mathrm{flip}} = \left( \begin{smallmatrix} 0&1\\ 1&0 \end{smallmatrix} \right)$
is a map over a binary space $X$. It is not QE because it has negative determinant. However, it \emph{is} QE with one hidden state.

This follows by constructing a special set of three LR matrices, each defined over a space $Y$ that has three elements,
such that the restriction to the first two elements of the product of those matrices is the bit flip:
\begin{equation}
\begin{bmatrix}
0 & 1 & 0\\
1 & 0 & 1\\
0 & 0 & 0
\end{bmatrix} =  \begin{bmatrix}
         1  &  0 & 0\\
         0 & 1 & 1\\ 
         0 & 0 & 0
        \end{bmatrix}
        \begin{bmatrix}
             1  &  1 & 0\\
             0 & 0 & 0\\ 
             0 & 0 & 1
            \end{bmatrix}
            \begin{bmatrix}
                 0  &  0 & 0\\
                 0 & 1 & 0\\ 
                 1 & 0 & 1
                \end{bmatrix}.
\label{eq:matrix_product}
\end{equation}
The restriction of the matrix on the LHS to its first two elements (i.e., the upper left block)
is the bit flip operation, by inspection. In addition, the first two matrices on the RHS are LR, by inspection.
The third matrix on the RHS is LR as well, but to see that we need to relabel the elements of $Y$ in such a way
that that matrix becomes block-diagonal. Specifically,
if we permute the second and the third elements of $Y$, we transform the third matrix as:
\[
\begin{bmatrix}
0 & 0 & 0\\
0 & 1 & 0\\
1 & 0 & 1
\end{bmatrix} 
\to
\begin{bmatrix}
         0 & 0 & 0\\
         1 & 1 & 0\\ 
         0 & 0 & 1
         \end{bmatrix} 
\]
which confirms that the rightmost of the three matrices is LR. 

So the RHS is a product of three LR matrices, and therefore
the LHS must be QE. This confirms that the bit flip is QE with one hidden state.

As an aside, note that the matrix on the LHS is not itself LR, even though it is a product of three LR matrices.
(This follows by verifying that there is no relabeling of the elements of $Y$
that changes the matrix on the LHS of \cref{eq:matrix_product}
into block-diagonal form.)
\label{ex:bit_flip}
\end{example}

We can easily generalize this technique to establish that any transposition is 
QE with one hidden state. Since any permutation can be written as a product of transpositions, and since the product of QE matrices
is QE, the next result follows immediately:

\begin{proposition}
\label{prop:single-valued-permutation}
Any permutation is QE with one hidden state.
\end{proposition}
\noindent
Together \cref{corr:QE-is-singular} and \cref{prop:single-valued-permutation} imply that any permutation matrix requires exactly one hidden state to be QE.

As mentioned in \cref{sec:introduction}, it is possible to
extend our analysis of single-valued maps to countably infinite spaces $X$, for
suitable extensions of our definitions. 
\cref{app:infinite} introduces one such extension
of our definitions, and then proves that with those extended definitions,
any single-valued function over a countably infinite $X$ can be 
implemented with at most one hidden state.

\subsection{General case}

We now present our main result for arbitrary stochastic matrices. To begin, recall that the
\emph{nonnegative rank} of an $n\times n$ stochastic matrix $P$ is the smallest $m$ such that $P$ can be written as $P= RS$, where $R$ is an $n \times m$ stochastic matrix and $S$ is an $m \times n$ stochastic matrix~\cite{cohen_nonnegative_1993}. Roughly speaking, the nonnegative rank of a matrix
is analogous to the number of independent mixing components in a mixture distribution.

\begin{restatable}{theorem}{nonnegativernakthm}
\label{thm:nonnegative_rank}
An $n \times n$ stochastic matrix $P$ with nonnegative rank $\nnr$ is quasistatically embeddable with $\nnr-1$ hidden states.
\end{restatable}

To understand this bound intuitively, suppose we have decomposed a given 
$n \times n$ stochastic matrix $P$ into a product of two rectangular
stochastic matrices of dimensions $n \times \nnr$ and $\nnr \times n$. These 
two rectangular matrices can be interpreted as representing transfers of probability between disjoint sets of states: the first matrix transfers probability from
the $n$ original states to\ $\nnr$ hidden ones, and the second map transfers
probability back to the $n$ original states. These
transfers of probability can be shown to be QE, establishing that no
more than $r$ hidden states are needed to implement $P$ in a QE manner. 

As we show in \cref{app:mainresultnnr}, it is possible to reduce this number of hidden states needed by one, which yields \cref{thm:nonnegative_rank}.

Since the nonnegative rank of an $n \times n$ matrix is  $n$ or less, \cref{thm:nonnegative_rank}
implies that any $n \times n$ stochastic matrix is QE with at most $n-1$ hidden states.
This is fewer than the number of
hidden states provided by adding a single independent, extra binary degree of freedom to
the system, since adding such a bit 
doubles the size of the state space. Thus, simply by adding a hidden bit to a system, we can implement any stochastic matrix in a thermodynamically reversible manner (presuming we
have freedom to use arbitrary sequences of LR matrices to implement
the stochastic matrix). 

Recall from \cref{prop:EPbound} that some stochastic matrices $P$ cannot be directly
implemented with zero EP, even if they are embeddable. \cref{thm:nonnegative_rank} tells us that
it is sometimes possible to reduce entropy production 
of implementing an embeddable $P$ by adding hidden states. This is illustrated with the following example.

\begin{example}
\label{ex:extraheat}
The ``partial" bit flip
\[P = \begin{bmatrix}
2/3 & 1/3 \\
1/3 & 2/3
\end{bmatrix}\]
is embeddable \cite{kingman_imbedding_1962} but has nonzero determinant ($\det P = 1/3$), and so cannot be QE by \cref{corr:QE-is-singular}. Specifically, for any initial probability distribution $p(0)$ (except possibly the fixed point of $P$,  $p(0) = \left(1/2, 1/2\right)^\intercal$), there is some unavoidable EP.  For example, by \cref{prop:EPbound}, for the initial distribution $p(0) = \left(1, 0\right)^\intercal$,
\[\Sigma \ge -\frac{4/9}{2 \ln 1/3}\approx 0.2 \,.\]

However, by \cref{thm:nonnegative_rank}, $P$ is QE with one hidden state. Thus, by adding a hidden state, the partial bit flip can be carried out while achieving zero EP. Note that since the change in the Shannon entropy $S(Pp(0))-S(p(0))$ is independent of how $P$ is implemented, the reduction in EP realized when implementing $P$ using a hidden state implies a decrease in the amount of entropy produced in the environment (e.g., as a reduction in generated heat). 
\end{example} 

As a final comment, while the nonnegative rank provides a general bound on
the minimal number of hidden
states needed to quasistatically embed a given matrix, other properties of the matrix can sometimes be used to further reduce the bound.
As a simple illustration, suppose that we know the nonnegative rank of some
stochastic matrix $P$---but also know that $P$ is block diagonal (up to a rearrangement of rows/columns by some permutation matrix), and that the greatest nonnegative
rank of any of the blocks is $k$.
Now we can implement $P$ by implementing each block in $P$ independently, one after the other. 
Implementing $P$ this way would allow us to repeatedly ``reuse'' whatever hidden states we have, 
for each successive implementation of a block. This means that $P$ is quasistatically 
embeddable using $k - 1$ hidden states. This is true even if the nonnegative rank of $P$
is (substantially) larger than $k$, in which case the direct application of \cref{thm:nonnegative_rank} would give a much weaker bound on the minimal required
number of hidden states.

\subsection{Coarse-grained states}

When introducing our notion of hidden states, we indicated that our upper bounds would also apply if the stochastic matrix $P$ described the evolution of probability over some number $m$ of coarse-grained ``macrostates" which collectively have internal to them $n$ ``microstates". To see this, it suffices to notice that a single local relaxation could send all the microstates associated to a macrostate to one of its microstates. This concentrates probability in as many microstates as there are macrostates, and the resulting ``empty" $m-n$ microstates can now be used as hidden states in the sense of our definitions and constructions.

For example, consider a system with three microstates $\{a,b,c\}$ which are grouped into two macrostates $0=\{a\}$ and $1=\{b,c\}$, representing the two states of a bit, that we wish to flip. This can be done by first applying the local relaxation $\{a\} \mapsto a, \{b,c\} \mapsto b$, leaving $c$ with no probability. Now, by carrying out the bit flip operation on microstates $a$ and $b$ using $c$ as a hidden state (as shown in Example \ref{ex:bit_flip}), we carry out a bit flip over the macrostates $0$ and $1$.

\section{Discussion}
\label{sec:disc}

\def\arraystretch{1.3}
\begin{table}
	\begin{tabular}{|l|c|c|c|}
	\hline 
	\textbf{Type of matrix $P$} & \multicolumn{3}{c|}{\textbf{Min. \# hidden states}}\tabularnewline
	\hline 
	\hline 
	\multicolumn{4}{|l|}{\emph{Single-valued}}\tabularnewline
	\hline 
	$\;$Invertible ($\ne I$) & \multicolumn{3}{c|}{1}\tabularnewline
	\hline 
	$\;$Non-invertible & \multicolumn{3}{c|}{0}\tabularnewline
	\hline 
	\shortstack{\emph{Noisy} \\ \;} & \multicolumn{2}{c|}{\shortstack{\vspace{0.25pt} \\ Lower bound\\(LE)\quad(QE)\\ \;}} & \shortstack{\vspace{0.25pt} \\ Upper bound\\(QE) \\ \; }\tabularnewline
	\hline 
	$\;$$\text{det}\,P=0$ & \multirow{2}{*}{\quad0\quad} & 0 & \multirow{4}{*}{$\nnr-1$}\tabularnewline
	\cline{1-1} \cline{3-3} 
	$\;$$\text{det}\,P>0$ &  & 1 & \tabularnewline
	\cline{1-3} 
	$\;$$\text{det}\,P<0$ & \multirow{2}{*}{1} & \multirow{2}{*}{1} & \tabularnewline
	\cline{1-1} 
	$\;$$\text{det}\,P>\prod_{i}P_{ii}$ &  &  & \tabularnewline
	 	\cline{1-4} 
	 	$\;$Limit Embeddable & 0 & 0 & 1 \tabularnewline
	 	\hline 
\end{tabular}

	\caption{Minimum number of hidden states required to limit-embed  {[}LE{]} or quasistatically embed 	{[}QE{]} a given stochastic matrix $P$. $\nnr$ is the nonnegative rank of $P$.\label{tab:known-bounds}}
\end{table}

In this paper, we consider implementing stochastic matrices $P$ using a master
equation. For some $P$, this is not possible. 
However, we show that it is always possible to implement any $P$ (to arbitrary
accuracy), if we have sufficiently many 
hidden states, in addition to the visible states that $P$ works on. 
We also show that it is always possible not only to implement $P$ (to arbitrary
accuracy), but to do so in a thermodynamically reversible manner by using hidden states. 
The minimal number of required hidden states for such thermodynamically reversible implementation
of a given stochastic matrix $P$
is a novel and fundamental kind of ``state space cost'' of carrying out  $P$.

We go on to derive
some bounds on the minimal number of such hidden states required for any given
$P$, either just to implement it, or to do so in a thermodynamically
reversible manner. In particular, we derive bounds on this minimal 
number of states needed for these two kinds of implementation, stated in terms of 
basic properties of $P$.

Our results for different kinds of $P$ are 
summarized in \cref{tab:known-bounds}, which lists upper and lower bounds for both limit-embedding
(i.e., implementing to arbitrary accuracy) and quasistatic embedding (i.e., implementing to
arbitrary accuracy in a thermodynamically reversible manner) for different classes of matrices. Any QE matrix is limit-embeddable, so the upper bounds established for quasistatic embedding also apply to limit-embedding. All the lower bounds in this table are tight, in the sense that in each class of matrices listed, there is a matrix that satisfies the lower bound.

\subsection{Interpretation of hidden states}

Often in physics there are some states of a system or even entire physical variables
that are hard to observe and control experimentally. Such ``hidden states'' are often
a problem to be circumvented, e.g., by coarse-graining. Whatever
method we use for dealing with hidden states invariably
affects the predictions we make (e.g., coarse-graining can
result in entropy increasing with time). However, such methods can allow us to proceed in an analysis despite our incomplete knowledge. 
 
At other times, the presence of hidden states can be a crucial property of a system, necessary for us to use a master equation to model the dynamics. Our results highlight the extent to which this is the case in different situations. In some cases, hidden states cannot be ignored---an engineer designing a physical system to implement a given $P$ must ensure that there is appropriate dynamic coupling between the hidden states and the visible ones that $P$ operates on, either implicitly or explicitly. In a different context, the number
of hidden states are a measure of how many internal states are being overlooked by a scientist who notices that some naturally
occurring system evolves (over discrete time) according to $P$, and
wants to presume that there is a master equation dynamics underlying that
evolution.

It is important to emphasize that the role for hidden states 
uncovered in our analysis
is different from the role of the states of the history tape in Bennett's reversible computation construction~\cite{bennett1982thermodynamics,bennett1973logical,levine1990note}, or the states of the extra bits in reversible Toffoli gates~\cite{fredkin1982conservative}. Like hidden states,
the states used in reversible computation ``facilitate'' the desired dynamics $P$, 
in a broad sense, without being in the space that $P$ runs over. However, the role 
of those states in reversible computation is to allow {logical} reversibility
of the conditional distributions giving the update rules, and therefore to lower the minimal thermodynamic work (rather than the entropy production) required to perform certain computations. In contrast, the conditional distributions we construct here
are made from LR matrices and are not logically reversible in general.

Indeed, the state space cost of a conditional distribution $P$ in some ways behaves in a manner
``opposite'' to the costs of $P$ discussed in the early thermodynamics of computation literature. A logically reversible function needs a hidden state to be implementable by a CTMC, while
a function that is logically irreversible does not. Therefore, as far as the number of hidden states is concerned, there are advantages
to being logically \emph{non}-invertible rather than being logically invertible.

\subsection{Biochemical oscillations}

One possible application of our results is to recent stochastic thermodynamics analyses of 
``Brownian clocks''~\cite{barato2016cost, barato2017coherence}. These
are biochemical oscillations in which a component (e.g. a protein) undergoes Markovian transitions (governed by some rate matrix $M(t)$) through a cycle of discrete states.\footnote{Note that the rate equation of a monomolecular chemical reaction network is also of the form of Eq.\eqref{eq:ctmc_def}, so the same mathematics can arise with a different interpretation.} 

Such clock-like oscillations are forbidden at equilibrium. However, they can be sustained out of equilibrium two ways: By fixed driving forces ($M(t) = M$ is constant but not detailed balanced), or by periodic driving (time-dependent variation of $M(t)$). 

In the former case, oscillations invariably dephase, though coherence can be preserved over many periods if the system is driven strongly and has many states \cite{barato2017coherence}. 
Even in the latter, more general case, it is challenging to preserve phase information. 
In particular, if we try to design a Brownian clock that remembers its phase exactly, we run into precisely the embedding problem discussed in this paper. For example, a clock with three states that remembers its phase (perfectly) must transform its own states according to the cyclic permutation:
\[
P_\text{tick} = \begin{bmatrix}
0 & 0 & 1\\
1 & 0 & 0\\
0 & 1 & 0
\end{bmatrix}
\]
each time it ticks. As we observe in \cref{sec:introduction}, this is impossible (since the product of the diagonal entries of $P_\text{tick}$ is less than its determinant \cite{goodman_intrinsic_1970, lencastre2016empirical}). In fact, we cannot even approximate these
clock dynamics by allowing an arbitrarily small amount of error. For example, a ``lazy" clock $P_\text{lazy} = (1-\epsilon)P_\text{tick} + \epsilon I$, which fails to tick with probability $\epsilon$ and slowly dephases, is not embeddable.

However, as we show above, all permutations, including $P_\text{tick}$, are QE with one hidden state, which means that if there is a hidden state in the clock (which may be occupied between ticks) it can be arbitrarily precise and remember its phase, while producing arbitrarily little entropy. Our more general results (e.g.~Theorem \ref{thm:nonnegative_rank}) say that periodic driving can be used to implement any variant of the stochastic matrix $P_\text{tick}$ (e.g., modeling a three state clock that loses coherence in a specific way or rate), as long as at least two hidden states are available.

\subsection{Future technical directions}
\label{sec:future_work}

Many of our results establish upper bounds on the number of hidden states required to quasistatically embed a given stochastic matrix $P$. We generally do not know how tight these upper bounds are, even in the worst cases.  In fact, we have no example of a matrix which requires more than one hidden state for limit-embedding or quasistatic embedding. One strategy to prove lower bounds would be to show that all QE matrices can be written as products of local relaxations, and then improve our arguments to establish that the number of hidden states used in our constructions (or modifications thereof) is minimal. This is a natural target for future work.

We are also keen to explore the potentially fruitful connections between the classical questions we have tackled in this paper and questions arising in the study of open quantum systems. The state of such systems can be described by a density matrix, and the transformation of state between two times is given by a completely positive trace-preserving map, or ``quantum channel", which is the quantum analogue of a stochastic matrix. The analogue of the classical master equation in this context is the Lindblad master equation \cite{lindblad1976}, and the analogue of the embedding problem is known as the Markovianity problem \cite{cubitt2012complexity}. Prior work related to our concerns in this paper include \cite{schnell}, which studies the relationship between time-homogeneous and inhomogeneous embedding in the quantum context, as well the results in \cite{hush} on mapping certain non-Lindblad master equations to systems obeying a Lindblad master equation coupled to a single extra ``ancillary" qubit.

Another possible direction involves pursuing continuous space generalizations of our results. Although master equations over a finite number of states are well-suited to represent many systems at some scales, continuous state spaces unavoidably appear in microscopic, classical descriptions. One way to handle this would be to introduce a fine discretization of space (see e.g. \cite{gingrich2017}) to turn say, diffusion in a compact region into a finite-state master equation. Although our results presented here would apply to such an approximation, they would not be addressing the natural questions for such a context, because they ignore the continuous structure of state space. For example, we have assumed throughout that the entries of $M(t)$ can be controlled freely and entirely independently of each other. In a discretized model of a continuous system, this would be like supposing that force fields or diffusivities could vary arbitrarily over the (very small) discretization scale. As an alternative, one could consider the embedding problem for Fokker-Planck equations, characterizing the set of conditional probability density functions that can be implemented with Fokker-Planck equations, and seeing whether adding additional dimensions can enlarge this set.

Of course, there is also an ``intermediate'' regime, in which the state space is infinite,
but still countable. We present a very preliminary
investigation of that regime in \cref{app:infinite}. Some possibly fruitful future work would involve 
extending the investigations there, e.g., to apply to noisy $P$ as well as single-valued
$P$. It might also be fruitful to consider alternatives to the ways that 
the analysis in \cref{app:infinite}
extends our basic concepts of ``implementing $P$'' from the finite state space case to the infinite
state space case. 

Even in naturally discrete systems, our ability to independently control transition rates between states could be constrained, which could give rise to state space costs greater than the very modest upper bound (one hidden bit) we establish here for any matrix $P$, and deserve further exploration.

Our proofs involve implementing a given conditional distribution $P$ via a sequence of discrete steps---local relaxations.  
The number of such discrete steps required to implement a given $P$ presents an interesting way of measuring the difficulty of physically implementing a stochastic matrix.
Analyzing the {number} of steps required to implement a given $P$, and how this number depends on the number of hidden states and the details of $P$, is an important direction for future work.
As a initial result in this direction, we note that our construction in the proof of \cref{thm:nonnegative_rank} requires $4\nnr n - 10$ local relaxations to be performed in sequence.
Such investigations are closely linked to the state space size versus timestep tradeoff for single-valued stochastic matrices, which we analyze in a companion paper~\cite{timestepspaper}. (See also work on a related question for the special case $n = 3$ in \cite{johansen1979bang}.)

\

\begin{acknowledgments}
We would like to thank Max Tegmark for useful discussions, and the Santa Fe Institute for helping to support this research. 
This paper was made possible through the support of Grant No.
FQXi-RFP-1622 from the FQXi foundation, and Grant No. CHE-1648973 from the U.S. National Science Foundation. 
\end{acknowledgments}

\bibliographystyle{unsrt}

\begin{thebibliography}{10}

\bibitem{van2013stochastic}
Christian Van~den Broeck et~al.
\newblock Stochastic thermodynamics: A brief introduction.
\newblock {\em Physics of Complex Colloids}, 184:155--193, 2013.

\bibitem{van_den_broeck_ensemble_2015}
C.~Van~den Broeck and M.~Esposito.
\newblock Ensemble and trajectory thermodynamics: {A} brief introduction.
\newblock {\em Physica A: Statistical Mechanics and its Applications},
  418:6--16, January 2015.

\bibitem{seifert2012stochastic}
Udo Seifert.
\newblock Stochastic thermodynamics, fluctuation theorems and molecular
  machines.
\newblock {\em Reports on Progress in Physics}, 75(12):126001, 2012.

\bibitem{van_kampen_stochastic_1981}
N.~G. Van~Kampen.
\newblock Stochastic processes in chemistry and physics.
\newblock {\em Amsterdam: North Holland}, 1:120--127, 1981.

\bibitem{givon2004extracting}
Dror Givon, Raz Kupferman, and Andrew Stuart.
\newblock Extracting macroscopic dynamics: model problems and algorithms.
\newblock {\em Nonlinearity}, 17(6):R55, 2004.

\bibitem{gaspard2006hamiltonian}
Pierre Gaspard.
\newblock Hamiltonian dynamics, nanosystems, and nonequilibrium statistical
  mechanics.
\newblock {\em Physica A: Statistical Mechanics and its Applications},
  369(1):201--246, 2006.

\bibitem{esposito2010finite}
Massimiliano Esposito, Ryoichi Kawai, Katja Lindenberg, and Christian Van~den
  Broeck.
\newblock Finite-time thermodynamics for a single-level quantum dot.
\newblock {\em EPL (Europhysics Letters)}, 89(2):20003, 2010.

\bibitem{elfving_zur_1937}
G~Elfving.
\newblock Zur theorie der {Markoffschen} ketten.
\newblock {\em Acta Soc SciFennicae n Ser A2}, 1937.

\bibitem{kingman_imbedding_1962}
J.~F.~C. Kingman.
\newblock The imbedding problem for finite {Markov} chains.
\newblock {\em Zeitschrift f{\"u}r Wahrscheinlichkeitstheorie und Verwandte
  Gebiete}, 1(1):14--24, January 1962.

\bibitem{frydman_total_1979}
Halina Frydman and Burton Singer.
\newblock Total positivity and the embedding problem for {Markov} chains.
\newblock {\em Mathematical Proceedings of the Cambridge Philosophical
  Society}, 86(02):339--344, September 1979.

\bibitem{fuglede_imbedding_1988}
B.~Fuglede.
\newblock On the imbedding problem for stochastic and doubly stochastic
  matrices.
\newblock {\em Probability Theory and Related Fields}, 80(2):241--260, December
  1988.

\bibitem{lencastre2016empirical}
Pedro Lencastre, Frank Raischel, Tim Rogers, and Pedro~G Lind.
\newblock From empirical data to time-inhomogeneous continuous markov
  processes.
\newblock {\em Physical Review E}, 93(3):032135, 2016.

\bibitem{jia2016solution}
Chen Jia.
\newblock A solution to the reversible embedding problem for finite markov
  chains.
\newblock {\em Statistics \& Probability Letters}, 116:122--130, 2016.

\bibitem{cubitt2012complexity}
Toby~S Cubitt, Jens Eisert, and Michael~M Wolf.
\newblock The complexity of relating quantum channels to master equations.
\newblock {\em Communications in Mathematical Physics}, 310(2):383--418, 2012.

\bibitem{bausch2016complexity}
Johannes Bausch and Toby Cubitt.
\newblock The complexity of divisibility.
\newblock {\em Linear algebra and its applications}, 504:64--107, 2016.

\bibitem{timestepspaper}
David~H. Wolpert, Artemy Kolchinsky, and Jeremy~A. Owen.
\newblock A space--time tradeoff for implementing a function with master equation dynamics.
\newblock {\em Nature Communications}, 10(1):1727, 2019

\bibitem{kolchinsky2016dependence}
Artemy Kolchinsky and David~H Wolpert.
\newblock Dependence of dissipation on the initial distribution over states.
\newblock {\em Journal of Statistical Mechanics: Theory and Experiment}, page
  083202, 2017.

\bibitem{turgut_relations_2009}
S.~Turgut.
\newblock Relations between entropies produced in nondeterministic
  thermodynamic processes.
\newblock {\em Physical Review E}, 79(4):041102, April 2009.

\bibitem{maroney2009generalizing}
O.J.E. Maroney.
\newblock Generalizing landauer's principle.
\newblock {\em Physical Review E}, 79(3):031105, 2009.

\bibitem{wolpert2015extending}
David~H Wolpert.
\newblock Extending landauer's bound from bit erasure to arbitrary computation.
\newblock {\em arXiv preprint arXiv:1508.05319}, 2015.

\bibitem{wolpert_baez_entropy.2016}
David~H Wolpert.
\newblock The free energy requirements of biological organisms; implications
  for evolution.
\newblock {\em Entropy}, 18(4):138, 2016.

\bibitem{goodman_intrinsic_1970}
Gerald~S. Goodman.
\newblock An intrinsic time for non-stationary finite {Markov} chains.
\newblock {\em Probability Theory and Related Fields}, 16(3):165--180, 1970.

\bibitem{esposito2010three}
Massimiliano Esposito and Christian Van~den Broeck.
\newblock Three faces of the second law. i. master equation formulation.
\newblock {\em Physical Review E}, 82(1):011143, 2010.

\bibitem{landauer1961irreversibility}
Rolf Landauer.
\newblock Irreversibility and heat generation in the computing process.
\newblock {\em IBM journal of research and development}, 5(3):183--191, 1961.

\bibitem{bennett1982thermodynamics}
Charles~H Bennett.
\newblock The thermodynamics of computation---a review.
\newblock {\em International Journal of Theoretical Physics}, 21(12):905--940,
  1982.

\bibitem{sagawa2014thermodynamic}
Takahiro Sagawa.
\newblock Thermodynamic and logical reversibilities revisited.
\newblock {\em Journal of Statistical Mechanics: Theory and Experiment},
  2014(3):P03025, 2014.

\bibitem{salamon_thermodynamic_1983}
Peter Salamon and R.~Stephen Berry.
\newblock Thermodynamic length and dissipated availability.
\newblock {\em Physical Review Letters}, 51(13):1127, 1983.

\bibitem{andresen_thermodynamics_1984}
Bjarne Andresen, R.~Stephen Berry, Mary~Jo Ondrechen, and Peter Salamon.
\newblock Thermodynamics for processes in finite time.
\newblock {\em Accounts of Chemical Research}, 17(8):266--271, 1984.

\bibitem{zulkowski2014optimal}
Patrick~R Zulkowski and Michael~R DeWeese.
\newblock Optimal finite-time erasure of a classical bit.
\newblock {\em Physical Review E}, 89(5):052140, 2014.

\bibitem{Saito2018SpeedLimit}
Naoto Shiraishi, Ken Funo, and Keiji Saito.
\newblock Speed limit for classical stochastic processes.
\newblock {\em Phys. Rev. Lett.}, 121:070601, Aug 2018.

\bibitem{diana2013finite}
Giovanni Diana, G~Baris Bagci, and Massimiliano Esposito.
\newblock Finite-time erasing of information stored in fermionic bits.
\newblock {\em Physical Review E}, 87(1):012111, 2013.

\bibitem{howie_subsemigroup_1966}
John~M. Howie.
\newblock The subsemigroup generated by the idempotents of a full
  transformation semigroup.
\newblock {\em Journal of the London Mathematical Society}, 1(1):707--716,
  1966.

\bibitem{cohen_nonnegative_1993}
Joel~E. Cohen and Uriel~G. Rothblum.
\newblock Nonnegative ranks, decompositions, and factorizations of nonnegative
  matrices.
\newblock {\em Linear Algebra and its Applications}, 190:149--168, 1993.

\bibitem{bennett1973logical}
Charles~H Bennett.
\newblock Logical reversibility of computation.
\newblock {\em IBM journal of Research and Development}, 17(6):525--532, 1973.

\bibitem{levine1990note}
Robert~Y Levine and Alan~T Sherman.
\newblock A note on bennett's time-space tradeoff for reversible computation.
\newblock {\em SIAM Journal on Computing}, 19(4):673--677, 1990.

\bibitem{fredkin1982conservative}
Edward Fredkin and Tommaso Toffoli.
\newblock Conservative logic.
\newblock {\em International Journal of Theoretical Physics}, 21(3):219--253,
  1982.

\bibitem{barato2016cost}
Andre~C Barato and Udo Seifert.
\newblock Cost and precision of brownian clocks.
\newblock {\em Physical Review X}, 6(4):041053, 2016.

\bibitem{barato2017coherence}
Andre~C Barato and Udo Seifert.
\newblock Coherence of biochemical oscillations is bounded by driving force and
  network topology.
\newblock {\em Physical Review E}, 95(6):062409, 2017.

\bibitem{lindblad1976}
G.~Lindblad.
\newblock On the generators of quantum dynamical semigroups.
\newblock {\em Comm. Math. Phys.}, 48(2):119--130, 1976.

\bibitem{schnell}
Alexander Schnell, Andr{\'e} Eckardt, and Sergey Denisov.
\newblock Is there a floquet lindbladian?
\newblock {\em arXiv preprint arXiv:1809.11121}, 2018.

\bibitem{hush}
Michael~R. Hush, Igor Lesanovsky, and Juan~P. Garrahan.
\newblock Generic map from non-lindblad to lindblad master equations.
\newblock {\em Phys. Rev. A}, 91:032113, Mar 2015.

\bibitem{gingrich2017}
Todd~R Gingrich, Grant~M Rotskoff, and Jordan~M Horowitz.
\newblock Inferring dissipation from current fluctuations.
\newblock {\em Journal of Physics A: Mathematical and Theoretical},
  50(18):184004, 2017.

\bibitem{johansen1979bang}
S{\o}ren Johansen and Fred~L Ramsey.
\newblock A bang-bang representation for 3$\times$ 3 embeddable stochastic
  matrices.
\newblock {\em Probability Theory and Related Fields}, 47(1):107--118, 1979.

\bibitem{johansen1973central}
S{\o}ren Johansen.
\newblock A central limit theorem for finite semigroups and its application to
  the imbedding problem for finite state markov chains.
\newblock {\em Zeitschrift f{\"u}r Wahrscheinlichkeitstheorie und Verwandte
  Gebiete}, 26(3):171--190, 1973.

\end{thebibliography}

\appendix

\section{Proof of \cref{lemma:closure}}
\lemmaclosure*
\begin{proof}
	If $T(t,t')$ is the transition matrix associated with rate matrix $M(t)$ and $S(t,t')$ is the transition matrix associated with $N(t)$, then $S(1/2, 1)T(0, 1/2) = U(0,1)$, where $U(t,t')$ is associated with rate matrix:
	\begin{equation}
	L(t) = \begin{cases*}
	M(t) & if  $t \in [0, \frac12]$  \\ \nonumber
	N(t) & if $t \in (\frac12, 1]$
	\end{cases*}
	\end{equation} To complete the proof note that the entropy production over the whole process is the sum of entropy production over
	both subintervals.
\end{proof}

\section{Proof of \cref{prop:EPbound}}
\label{prop2proof}


\EPboundprop*
\begin{proof}
We begin with a result from \cite{Saito2018SpeedLimit}.  Eq.~14 in that paper, with minor rearranging, states
\begin{align}
\label{eq:saitores}
 \Sigma(p(0), M(t), [0,\tau]) \ge \frac{\lvert p(\tau) - p(0)\rvert^2_1}{2 \tau \langle  A \rangle_\tau} \,,
\end{align}
where $\langle A \rangle_\tau$ is the so-called ``dynamical activity'',
\begin{align*}
\langle A \rangle_\tau &:= \frac{1}{\tau} \int_0^\tau \sum_{i\ne j} M_{ij}(t) p_j(t) dt\\
&= -\frac{1}{\tau} \int_0^\tau \sum_{i} M_{ii}(t) p_i(t) dt \,.
\end{align*}
Since for all $i$, $p_i(t)\le 1$ and $M_{ii}(t)\le 0$, we can upper bound the dynamical activity as
\begin{align}
\langle A \rangle_\tau \le  -\frac{1}{\tau} \int_0^\tau \sum_{i} M_{ii}(t) dt = -\frac{1}{\tau} \ln(\det P) \,,
\label{eq:saitoAbound}
\end{align}
where we've used $P=T_M(0,\tau)$ and \cref{eq:detid}. The result follows by combining \cref{eq:saitores,eq:saitoAbound} and taking $p=p(0)$, $Pp = p(\tau)$.
\end{proof}

\cref{corr:QE-is-singular}, that any non-identity QE matrix $P$ is singular, follows directly from \cref{prop:EPbound}.  Choose $p$ to be any non-fixed-point of $P$. Then, in order for entropy production to be made arbitrarily small, as required by the definition of QE, it must be that $\det P = 0$.

\section{Proof of \cref{embedqe1}}
\label{app:embedqe1proof}
\embedqeprop*
\begin{proof}
Let $P$ be a limit-embeddable matrix. By definition, there is an embeddable matrix $P'$ within any desired distance $\epsilon_1$ of $P$.

Now consider a stochastic matrix which keeps every state fixed except state $i$, which it sends to $j$ with probability $\alpha$, and leaves 
alone with probability $1 - \alpha$. In control theory, such matrices are called \emph{Poisson matrices}. It is known that any embeddable matrix $P'$ can be approximated arbitrarily closely by a finite product of such matrices \cite{johansen1973central}.  Let $P''$ indicate a product of Poisson matrices which approximates $P'$ within a distance $\epsilon_2$.

A Poisson matrix is QE with one hidden state by \cref{twostate}.  Since $P''$ is a product of Poisson matrices, $P''$ is QE with one hidden states.
Thus, $P$ is arbitrarily close (less than $\epsilon_1$ + $\epsilon_2$) to $P''$, a matrix which is QE with one hidden state. 

Now choose such a $P''$ within $\epsilon/2$ of $P$.  If we wish to implement $P$ within $\epsilon$ while producing total entropy less than $\delta$, with one hidden state, we can do so with a rate matrix $M(t)$ that implements $P''$ to within $\epsilon/2$ while producing less than $\delta$ entropy, with one hidden state. There always is one since $P''$ is QE with one hidden state. This establishes the result.

\end{proof}

\section{Proof of \cref{prop:relaxQE}, that any LR matrix is QE}
\label{LRisQE}
\proprelaxQE*
\begin{proof}
We first show this for local relaxations that have a {single} block, that is, those of the form 
\[
P = \mathbf{\pi}\mathbf{1}^T \,,
\]
where $\pi$ is an $n$-dimensional column-vector with positive entries that sum to 1, and $\mathbf{1}$ is an $n$-dimensional column-vector of 1s.
We construct a sequence of CTMCs that approximates $P$ arbitrarily well while achieving arbitrarily small total EP for some initial distribution $q$.

Consider the rate matrix whose off-diagonal entries are given by $M_{ij}(t) = \alpha w_i(t)$, where $w_i(t) = (1-t)q_i + t \pi_i$. The associated master equation decouples,
\begin{align}
\dot{p_i}(t) = \alpha(w_i(t) - p_i(t)) \,,
\label{eq:deriv}
\end{align} 
and can be solved explicitly (using the initial condition $p(0) = q$),
\begin{align}
p_i(t) = w_i(t)+\alpha^{-1}(q_i-\pi_i)(1-e^{-\alpha t}) \,.
\label{eq:p-solution}
\end{align}
It is clear that by making $\alpha$ sufficiently large, this CTMC will transform any initial distribution at $t=0$ arbitrarily close to final distribution $\pi$ at $t=1$. That is,
\[
\lim_{\alpha \rightarrow \infty} p_i(1) =\lim_{\alpha \rightarrow \infty} w_i(1)= \pi_i(t) \,,
\]
Thus, it will approximate $P$ arbitrarily well. We must now show that by making $\alpha$ large we can also make the total EP as small we like.  

Since the rates $M_{ij}(t)$ satisfy detailed balance at all times, $M_{ij}(t) w_j(t) = M_{ji}(t) w_i(t)$, 
the total entropy production can be written as
\begin{align*}
\Sigma =  - \sum_i \int_0^1 \dot{p}_i(t) \ln \frac{p_i(t)}{w_i(t)}\; dt
\end{align*}
We split this integral into two parts, 
\begin{align*}
\int_0^1 \dot{p}_i(t) \ln p_i(t) \; dt - \int_0^1 \dot{p}_i(t) \ln w_i(t)\; dt
\end{align*}

The first integral can be evaluated as,
\begin{align}
\int_0^1 \dot{p}_i(t) \ln p_i(t) dt & = \int_0^1 \frac{d}{dt}\left[ p_i(t) \ln p_i(t) - p_i(t) \right] \; dt \nonumber \\
& = p_i(1) \ln p_i(1) - q_i \ln q_i + q_i - p_i(1) \,, \nonumber 
\end{align}
which, as $\alpha \rightarrow \infty$, converges to
\begin{align}
\pi_i \ln \pi_i - q_i \ln q_i +  q_i - \pi_i \,.
\label{eq:int1}
\end{align}

The second integral can be written as
\begin{align*}
\int_0^1 \dot{p}_i(t) \ln w_i(t)\; dt = (\pi_i-q_i) \int_0^1 (1-e^{-\alpha t}) \ln w_i(t) \; dt \,,
\end{align*}
where we've plugged \cref{eq:p-solution} into \cref{eq:deriv}.  As $\alpha \rightarrow \infty$, the integrand converges pointwise and monotonically to $\ln w_i(t)$, which is bounded for all $t\in(0,1)$. 
By the dominated convergence theorem, as $\alpha \rightarrow \infty$, the second integral converges to 
\begin{align}
& (\pi_i-q_i) \int_0^1 \ln w_i(t) \; dt \\
& =  (\pi_i-q_i) \int_0^1 \ln \left( (1-t) q_i + t \pi_i \right) \; dt \nonumber \\
& =  \pi_i \ln \pi_i - q_i \ln q_i + q_i - \pi_i
\,. \label{eq:int2}
\end{align}

Comparing \cref{eq:int1} and \cref{eq:int2}, we see that in the limit of $\alpha \rightarrow 0$, the two integrals cancel, and so $\Sigma \rightarrow 0$.

The general case of a local relaxation with \emph{multiple} blocks follows immediately, because a block diagonal matrix is QE if the blocks are---a block diagonal matrix $A$ with blocks $B_i$ can be factored as a product $A=(B_1\oplus I \oplus \cdots)(I\oplus B_2 \oplus \cdots)\cdots$.
\end{proof}

\section{Proof of main result in \cref{sec:embedding-with-aux-states}}
\label{app:mainresultnnr}

To establish our main result we first establish some preliminary results.

\begin{lemma}\label{twostate}
Any $2 \times 2$ stochastic matrix is QE with one hidden state. 
\end{lemma}

\begin{proof}
For any real number $p \in [0, 1]$,
define $\bar{p} \coloneqq 1-p$. Without loss of generality, write $P$ in matrix notation as
\[P = \begin{bmatrix}
     \bar{p}  &  q\\
     p & \bar{q}
    \end{bmatrix} \] 

We consider two cases. If $p \leq \bar{q}$, then if $x = p/\bar{q}$ and $y = q$, we have
\[
    \begin{bmatrix}
         \bar{p}  &  q & \bar{p}\\
         p & \bar{q} & p \\
         0 & 0 & 0
        \end{bmatrix} =     \begin{bmatrix}
                 1  &  0 & 1\\
                 0 & 1 & 0 \\
                 0 & 0 & 0
                \end{bmatrix}     \begin{bmatrix}
                         y  &  y & 0\\
                         \bar{y} & \bar{y} & 0 \\
                         0 & 0 & 1
                        \end{bmatrix}     \begin{bmatrix}
                                 x  &  0 & x\\
                                 0 & 1& 0 \\
                                 \bar{x} & 0 & \bar{x}
                                \end{bmatrix}. 
    \] If instead $p > \bar{q}$, then if $x = \bar{p}$, $y=\bar{q}/p$, we have 
        \[\begin{bmatrix}
             \bar{p}  &  q & q\\
             p & \bar{q} & \bar{q} \\
             0 & 0 & 0
            \end{bmatrix} =     \begin{bmatrix}
                     1  &  0 & 1\\
                     0 & 1 & 0 \\
                     0 & 0 & 0
                    \end{bmatrix}     \begin{bmatrix}
                             x  &  x & 0\\
                             \bar{x} & \bar{x} & 0 \\
                             0 & 0 & 1
                            \end{bmatrix}     \begin{bmatrix}
                                     1  &  0 & 0\\
                                     0 & y & y \\
                                     0 & \bar{y} & \bar{y}
                                    \end{bmatrix}. 
        \]
In both cases, the factors on the right hand side are local relaxations, so the result is established. 
\end{proof}

\subsection{Transfers of probability using a relay state}

In what follows, it will be useful to consider stochastic matrices of the form:
\[
\begin{bmatrix}
I & P\\
0 & D
\end{bmatrix} \quad \mathrm{ or } \quad \begin{bmatrix}
D & 0\\
P & I
\end{bmatrix}
\]
where $P$ is an $m \times p$ matrix with positive entries with column sums less than or equal to one, $D$ is a $p \times p$ diagonal matrix that makes the whole block matrix stochastic, and $I$ is the $m \times m$ identity matrix.

We call such matrices \textbf{transfers}, because they can be viewed as representing transfers of probability between two disjoint sets of states of sizes $p$ and $m$. One can show that:

\begin{proposition}
\label{prop:transfer}
Any transfer $T$ is QE with one hidden state.
\end{proposition}

We will make repeated use of the map, $\mathcal{R}^{i \rightarrow j}(\alpha)$,  acting on a set with two elements, which keeps state $j$ fixed, but sends state $i$ to $j$ with probability $\alpha$ and leaves 
it as $i$ with probability $1 - \alpha$. Such a matrix, which is called a \emph{Poisson matrix} in control theory, is QE with one hidden state by \cref{twostate}.

\begin{figure}
\centering
\includegraphics[width=0.5\textwidth]{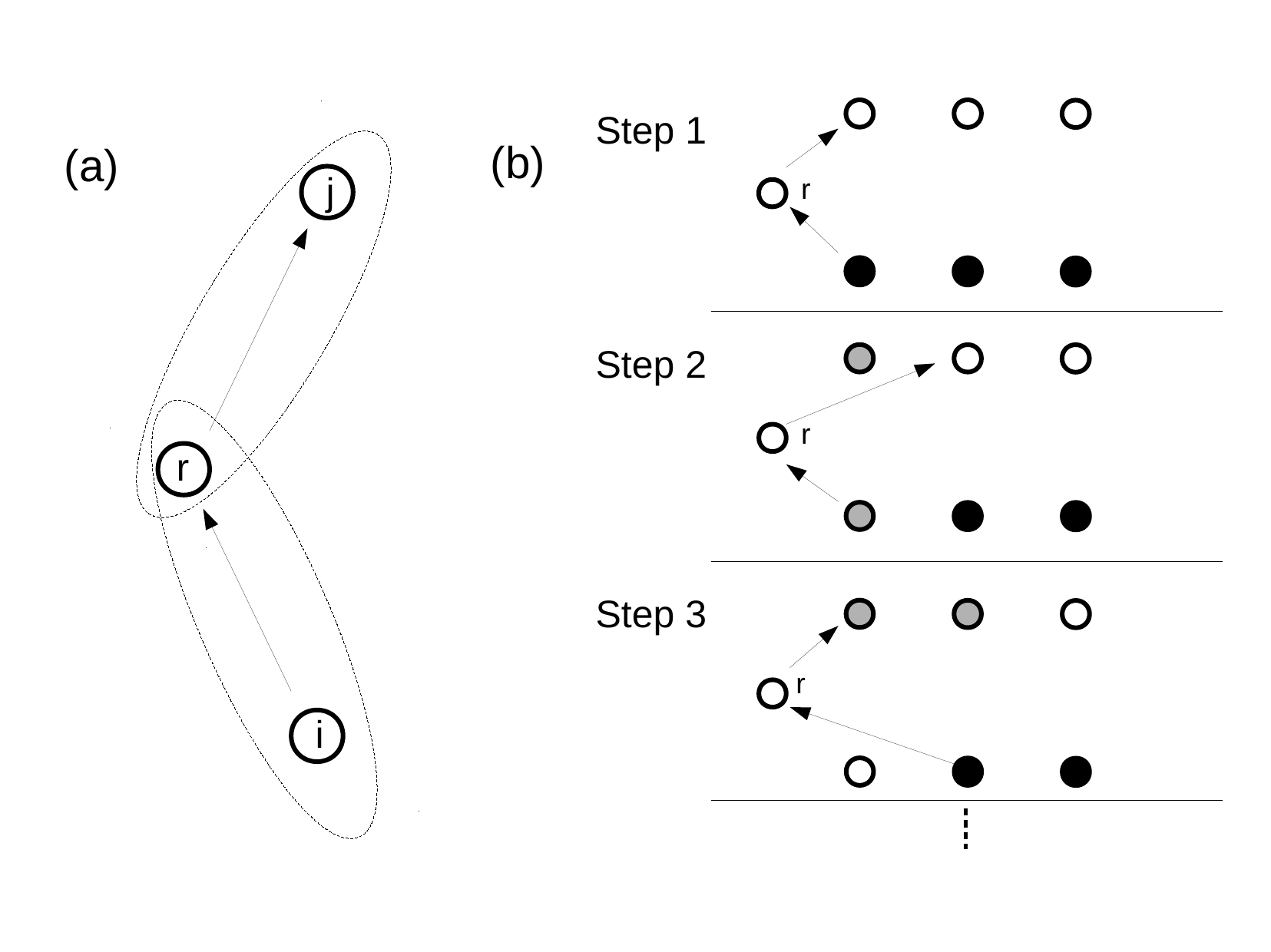}
\caption{(a) A composition of LRs can be used to effect a transfer of probability $\alpha$ from $i$ to $j$, using a relay state $r$. (b) A sequence of these operations can be composed to effect any transfer $T$.}
\label{fig:relaytrick}
\end{figure}
\begin{proof}[Proof of \cref{prop:transfer}]
Without loss of generality\footnote{The other kind of transfer can be rewritten in this form by permuting its rows and columns.}, let
\[
T = \begin{bmatrix}
D & 0\\
P & I
\end{bmatrix}
\]

If $P$ has only one nonzero row, whose entries are $p_i$, $i = 1, \dots, n$, then 
\[T= \prod_{i = 1}^n \mathcal{R}^{i \rightarrow j}(p_i)\]
where the indices $i$ and $j$ are now labeling rows (e.g.~states) in the larger matrix $T$. The first $n$ rows that $i$ ranges over become the rows of $D$, and $j$ is the index of the single nonzero row of $P$, now as a submatrix of $T$. So by construction $j > n \geq i$, and the maps $\mathcal{R}^{i \rightarrow j}(p_i)$ for different $i$ do not interfere with one another (they commute and fix each others images). 

Now suppose $P$ has two nonzero rows. Write $P$ as the sum of two matrices $P = P_1 + P_2$, each zero except for \emph{one} of the rows of $P$. Let $D_1$ be the diagonal matrix which makes
\[\begin{bmatrix} D_1 & 0 \\ P_1 & I\end{bmatrix}\]
stochastic, and write 
\[
\begin{bmatrix} D & 0 \\ P & I\end{bmatrix} = \begin{bmatrix} D_2D_1 & 0 \\ P_2 + P_1 & I\end{bmatrix} = \begin{bmatrix} D_2 & 0 \\ P_2D_1^{-1} & I\end{bmatrix} \begin{bmatrix} D_1 & 0 \\ P_1 & I\end{bmatrix}
\]
where $D_2$ is chosen to make the matrix is appears in stochastic (this is always possible since the column sums of $P$ are all less than 1). Since the product of stochastic matrices is stochastic, the matrix appearing on the left hand side is $T$, establishing the proposition for $P$ with two nonzero rows. The result for general $P$ follows by induction.

\end{proof}

\subsection{Minimum number of hidden states is less than nonnegative  rank}

\begin{lemma}
An $n \times n$ stochastic matrix $P$ with nonnegative  rank $2$ is QE with one hidden state.\label{lem:rank2}
\end{lemma} 

\begin{proof}
Write the nonnegative  rank decomposition $P = RS$ where $R$ and $S$ are stochastic matrices of dimensions $n \times 2$ and $2 \times n$, respectively. We further decompose $R$ and $S$ into sub-matrices: $R = \begin{bmatrix}
R_m \\
R_2
\end{bmatrix}$ and $S = [S_m \; S_2]$, where $R_2$ and $S_2$ are $2 \times 2$. Now note that
\[P = \begin{bmatrix}
R_m S_m & R_m S_2 \\
R_2S_m & R_2S_2
\end{bmatrix}
\]
and
\[
\begin{bmatrix}
R_m S_m & R_m S_2 \\
R_2 S_m & R_2 S_2
\end{bmatrix} = \begin{bmatrix}
I & 0 \\
0 & R_2D^{-1}
\end{bmatrix}\begin{bmatrix}
I & R_m \\
0 & D
\end{bmatrix}\begin{bmatrix}
0 & 0 \\
S_m & I
\end{bmatrix}
\begin{bmatrix}
I & 0 \\
0 & S_2
\end{bmatrix}
\]
where $D$ is the diagonal matrix that makes the second factor stochastic, and $I$ is an identity matrix of dimensions suitable to where it appears. To show that $P$ can be implemented with one hidden state, it suffices to show that each of the factors on the right hand side can be.

The middle two factors represent transfers, in the sense described above, so they can implemented with one hidden state. The first and last factors represent operations performed on just two states, so by \cref{twostate} can also be implemented with one hidden state. Note that $R_2 D^{-1}$ is stochastic, because the column sums of $R_2$ are exactly the diagonal entries of $D$ (recall $D$ was chosen to make the second factor stochastic). 
\end{proof}

We are now ready to prove \cref{thm:nonnegative_rank}.

\nonnegativernakthm*
\begin{proof} 
Suppose $\nnr > 2$ (if not, the result follows from \cref{lem:rank2}). As before, write $\pi = RS$ where $R$ and $S$ are stochastic of dimensions $n \times \nnr$ and $\nnr \times n$, respectively. But this time, decompose $R$ and $S$ differently: $R = [R_m \; P]$ and $S = \begin{bmatrix}
S_m \\
Q
\end{bmatrix}$, where $P$ is $n \times 2$ and $Q$ is $2 \times n$. Note that $\pi = R_m S_m + P Q$.
\[
\begin{bmatrix}
PQ+R_m S_m & R_m \\
0 & 0
\end{bmatrix} = \begin{bmatrix}
I & R_m \\
0 & 0
\end{bmatrix}\begin{bmatrix}
PQD^{-1} & 0 \\
0 & I
\end{bmatrix}\begin{bmatrix}
D & 0 \\
S_m & I
\end{bmatrix}
\]
where $D$ is the diagonal matrix that makes the second factor stochastic, and $I$ is an identity matrix of dimensions suitable to where it appears. 

To prove the Theorem, it suffices show that each of the $(n + k - 2) \times (n + k - 2) $ matrices on the right hand side care QE with one hidden state.

The first and last factors represent transfers from $\nnr-2$ hidden states to the original $n$ states (and vice versa), and can be implemented using one hidden state, as described earlier. The middle factor, which represents an operation performed only on the original states, has nonnegative rank 2 (note that $P$ and $QD^{-1}$ are stochastic), so by \cref{lem:rank2} can also be implemented with one hidden state. This establishes the result.
\end{proof}

\section{Countably infinite state spaces\label{app:infinite}}

We now consider the case where the state space of our system, $X$, is countably infinite, and restrict attention to the implementation of single-valued maps $f : X \to X$ over such spaces.

In our results above, we establish quasistatic embeddability of various matrices by writing them as products of local relaxations. In the infinite case, we can similarly ask what is possible by composing local relaxations, which in the context of single-valued maps 
are exactly the idempotent functions. 

One natural way to extend our analysis to the case of countably infinite $X$ is to consider  a sequence
of idempotent functions over some $Y \supset X$ that, when restricted to $X$, gives the desired $f$. To allow
the sequence of functions to implement $f(x)$ for \emph{all} $x \in X$, in general we must consider sequences
that are infinite. However, the infinite product of local relaxations need not be QE (or even be well-defined). To circumvent
this issue we impose a ``practical" interpretation of what it means to implement $f$,
by requiring that any particular input $x \in X$ is mapped to $f(x)$ after finitely many idempotents. 

Adopting this interpretation, in this appendix we establish the following result:
\begin{proposition}\label{prop:infinite}
	Let $X$ be a countable set, and take $Y := X \cup \{z\}$. For any function $f : X \to X$ there is a sequence $\{g_i\}$ of idempotent functions $g_i : Y \to Y$ such that for all $x \in X \subset Y$, there is an $m$ (which can depend on $x$) such that for all $r \geq m$, $(g_r \circ g_{r-1} \circ \cdots \circ g_1)(x) = f(x)$.
\end{proposition}

\noindent \cref{prop:infinite} is a kind of infinite analog of \cref{prop:single-valued-singular,prop:single-valued-permutation}, which showed that any function over a finite $X$ can be implemented with at most one hidden state.

To prove \cref{prop:infinite} we begin with the following lemma:

\begin{lemma}\label{lem:infinite}
Let $Y$ be $\mathbb{Z}$ with one point added. There is a sequence $\{g_i\}$ of idempotent functions $g_i : Y \to Y$ such that for all $n \in \mathbb{Z} \subset Y$, there is an $m$ (which can depend on $n$) such that for all $r \geq m$, $(g_r \circ g_{r-1} \circ \cdots \circ g_1)(n) = n + 1$.
\end{lemma}
\begin{proof}
First, let $[\alpha]$ indicate the cyclic permutation $(-1, 0, 1)$ and, for any integer $i \in \mathbb{Z}$, let $[i]$ indicate the cyclic permutation $(-i, i+1, -i-1)$.

For any particular $j \in \mathbb{Z}$, consider the following composition of permutations, 
\[
M^{(j)} := [j] \ldots [2][1][\alpha] \,.
\]
We show by induction that $M^{(j)}$ sends each state $i \in \{-j-1, \ldots, j\}$ to $i+1$ (while also sending $j+1$ to $-j-1$ as a `side effect').  First, it can be verified by inspection that it holds true for $j=0$ (when $M^{(j)} = [\alpha]$). Second, assume it holds true for $M^{(j)}$ and then observe that $$M^{(j+1)} = [j+1] M^{(j)} = (-j-1, j+2, -j-2) M^{(j)}\,.$$
So, applying $[j+1]$ after $M^{(j)}$ results in:
\begin{enumerate}
\item[(a)]  $-j-2$ being sent to $-j-1$;
\item[(b)] the `side effect' of  $M^{(j)}$ being fixed, in that, $M^{(j)}$ mapped $j+1 \mapsto -j-1$, but $[j+1]$ maps $-j-1\mapsto j+2$, so the combined result is $j+1 \mapsto j+2$.
\end{enumerate} 
Thus, $M^{(j+1)}$ sends each state $i \in \{-(j+1)-1, \ldots, (j+1)\}$ to $i+1$ (while now sending $(j+1)+1$ to $-(j+1)-1$ as a `side effect').










Note that a cyclic permutation like $[-j]$ and $[\alpha]$ affect only a finite number of points in $\mathbb{Z}$, and thus can be written as a product of a finite number of transpositions, which can in turn be written as a product of three idempotents that use one hidden state (e.g., as in \cref{ex:bit_flip}). This hidden state is provided by $Y$, which 
lets us construct a sequence of idempotents $\{g_i\}$ over $Y$ with the property we desire.
\end{proof}

We can now use a ``dovetailing'' algorithm to construct a sequence of local idempotents
that implements any specified function, thereby establishing \cref{prop:infinite}:

\begin{proof}[Proof of \cref{prop:infinite}]
Suppose first that $f$ is a bijection. Partition $X$ according to the orbits of $f$---that is, two points $x$ and $y$ are in the same part if there is some $n$ such that $y = f^n(x)$ or $x = f^n(y)$. 

There will be countably many orbits ${A_j}$. Each of the $A_j$ is either finite and $f$ restricted to it is a permutation, or else $A_j$ is countably infinite and there is a numbering of its elements (using positive and negative integers) such that $f$ restricted to that orbit is $n \to n+1$. In either case, using \cref{prop:single-valued-singular} or \cref{lem:infinite}, we can form for each orbit a sequence $\{(g_j)_i\}$ that implements the restriction of $f$ to that orbit. 
Extend the maps in each sequence to all of $Y$ by having them fix the elements in all other orbits.

Now form a new sequence by interleaving these (countably many) sequences in a way that preserves the order of elements in each individual sequence. For example,
\[
\{(g_1)_1, (g_1)_2, (g_2)_1, (g_1)_3, (g_2)_2, (g_3)_1, \dots \}.
\]
This sequence satisfies the requirements of the theorem. Any $x \in X$ is in some orbit $A_j$, and so there is some element of the associated sequence $(g_j)_m$ after the application of which $x$  is mapped to $f(x)$ and remains so. We use here the fact that the interleaving to form the larger sequence preserves order, and that idempotents that are members of sequences corresponding to the other orbits fix $A_j$.

If $f$ is not a bijection, consider the idempotent map $a$ that, for all $x \in X$, sends all elements in the inverse image $f^{-1}(x)$ to a distinguished element $w(x) \in f^{-1}(x)$. Write $Z = \operatorname{img} f \cup \operatorname{img} a$. We can write $f = h \circ a$, where $h: Z \to Z$ is a bijection. So using the construction above, we can make a sequence of idempotents for $h$ with the desired property. Adding $a$ to the beginning of the sequence forms a sequence that implements $f$. For each $x$ there exists an $w(x)$, and under the implementation of $h$,  $w(x) \mapsto f(w(x))$ once some number $m$ idempotents are applied. But note that $x$ was first mapped to $w(x)$, under the initial application of $a$. Thus, once $m$ idempotents of $h$ have been applied, the combined effect on $x$ is $x \mapsto w(x) \mapsto f(w(x))$, and $f(w(x))=f(x)$, by the definition of $w(x)$, so this completes the argument.
\end{proof}

\end{document}